\documentclass[11pt,a4paper]{article}
\usepackage{amsthm,amsmath}
\usepackage{amssymb,amstext,amsbsy,bbm,natbib}
\usepackage{color}
\usepackage{verbatim}
\usepackage{ bm}

\usepackage{times}
\let\cite=\citet
\usepackage{graphicx}
\newcommand{\Rlogo}{\protect\includegraphics[height=1.8ex,keepaspectratio]{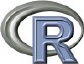}}

\newtheorem{thm}{Theorem}

\newcommand*{\nrm}[1]{\left\| #1 \right\|}  
\newcommand*{\abs}[1]{\left| #1 \right|} 
\newcommand{\E}{\mbox{$\mathbb{E}$}}

\newcommand{\PP}{\mbox{$\mathbb{P}$}}

\newcommand{\Rset}{\mathbb{R}}
\newcommand{\Nset}{\mathbb{N}}

\newcommand{\CF}{\mathcal{F}}

\newcommand\1{\leavevmode\hbox{\rm \small1\kern-0.35em\normalsize1}}
\newcommand\ind[1]{\1_{\{#1\}}}
\newcommand\egaldef{\stackrel{\mbox{\upshape\tiny def}}{=}}

\begin{document}

\title{A fast and recursive algorithm for clustering large datasets with $k$-medians}
\author{Herv\'e Cardot\footnote{Corresponding author.}$^{\ (a)}$, Peggy C\'enac$^{(a)}$ and Jean-Marie Monnez$^{(b)}$ \\
(a) Institut de Math\'ematiques de Bourgogne, UMR 5584, Universit\'e de Bourgogne, \\ 9 Avenue Alain Savary, 21078 Dijon, France \\
(b) Institut Elie Cartan, UMR 7502, Nancy Universit\'e, CNRS, INRIA, \\ B.P. 239-F 54506 Vandoeuvre l\`{e}s Nancy Cedex, France}

\maketitle

\begin{abstract}
Clustering with fast algorithms large samples of high dimensional data is an important challenge  in computational statistics. Borrowing ideas from \cite{MacQueen} who introduced a sequential version of the $k$-means algorithm,   a new class of recursive stochastic gradient algorithms  designed for the $k$-medians loss criterion is proposed. By their recursive nature, these algorithms are very fast and are well adapted to deal with large samples of data that are allowed to arrive sequentially. It is proved that the stochastic gradient algorithm converges  almost surely to the set of stationary points of the underlying loss criterion. A particular attention is paid to the averaged versions,  which are known to have better performances, and a data-driven procedure that allows automatic selection of the value of the descent step is proposed.
 The performance of the averaged sequential estimator is compared on a simulation study, both in terms of computation speed and accuracy of the estimations, with more classical partitioning techniques  such as $k$-means, trimmed $k$-means and PAM (partitioning around medoids).
Finally, this new online clustering technique is illustrated on determining television audience profiles with a sample of more than 5000 individual television audiences measured every minute over a period of 24 hours.
\end{abstract}


\section{Introduction}\label{sec:intro}

Clustering with fast algorithms large samples of high dimensional data is an important challenge  in computational statistics and  machine learning, with applications in various domains such as image analysis, biology or computer vision. There is a vast literature on clustering techniques and  recent discussions and reviews may be found in  \cite{JMF99} or \cite{GanMaWu2007}. 
Moreover, as argued in \cite{Bot10}, the development of  fast  algorithms is even more crucial when the computation time is limited and the sample is potentially very large, since fast procedures will be able  to deal with a larger number of observations and will finally provide better estimates than slower ones. 

\medskip

We focus here on partitioning  techniques which  are able to deal with large samples of data, assuming  the number $k$ of clusters is fixed in advance.  
The most popular clustering methods  are  probably the non sequential (\cite{Forgy65}) and the sequential (\cite{MacQueen}) versions of the $k$-means algorithms.
They are very fast and only require $O(kn)$ operations, where $n$ is the sample size. They aim at finding local minima of  a quadratic criterion and the cluster centers are given by the barycenters of the elements belonging to each cluster. 
A  major drawback of the $k$-means algorithms is that they are based on  mean values and, consequently, are very sensitive to  outliers. 
Such atypical values, which may not be uncommon in  large samples,    can deteriorate significantly the performances of these algorithms, even if they only represent a small fraction of the data as explained in \cite{GEGMM10} or \cite{CGVZ07}.
The $k$-medians approach is  a first attempt to get more robust clustering algorithms; it was suggested by \cite{MacQueen}  and developed by \cite{KauRou90}. It  consists in considering criteria based on least norms instead of least squared norms, so that the cluster centers are the spatial medians, also called geometric or $L_1$-medians (see \cite{Sma90}), of the elements belonging to each cluster. Note that it has been proved in \cite{Laloe2010} that under general assumptions, the minimum of the objective function is unique.  Many algorithms have been proposed in the literature to find this minimum. The most popular one is certainly  the PAM (partitioning around medoids) algorithm which has been developed by \cite{KauRou90} in order to search for local minima  among the elements of the sample. Its computation time  is $O(kn^2)$ and as a consequence,  it is not very well adapted for large sample sizes.  Many strategies have been suggested in the literature to reduce the computation time of this algorithm. For example subsampling (see \textit{e.g} the algorithm CLARA in \cite{KauRou90} and  the algorithm CLARANS in \cite{clarans02}), local distances computation (\cite{ZCoul2005}) or the use of weighted distances during the iteration steps (\cite{ParkJun2008}), allow  one to reduce significantly the computation time without deteriorating the accuracy of the estimated partition.
%  at the expense of a decreasing global accuracy. 

Trimmed $k$-means (see \cite{GGMM2008, GEGMM10} and references therein) is also a popular modification of the $k$-means algorithm that is more robust (see \cite{GG99}) in the sense that it has a strictly positive breakdown point, which is not the case for the $k$-medians.  Note however that the breakdown point  is a pessimistic indicator of robustness since it is based on the worst possible scenario. For a small fraction of outliers whose distance is moderate to the cluster centers, $k$-medians remain still competitive compared to trimmed $k$-means as seen in the simulation study. Furthermore, from a computational point of view,  performing trimmed $k$-means needs to sort the data  and this step requires $O(n^2)$ operations, in the worst cases, at each iteration so that its execution time can get large when one has to  deal with large samples.

\medskip
  
Borrowing ideas from  \cite{MacQueen} and \cite{hart75} who have first proposed sequential clustering algorithms  and \cite{CCZ10} who have studied the properties of stochastic gradient algorithms that can give efficient recursive estimators of the geometric median in high dimensional spaces, we propose in this paper a recursive strategy that is able to estimate the cluster centers by minimizing a $k$-medians type criterion.
One of the main advantages of our approach, compared to previous ones, is that it can be computed in only $O(kn)$ operations so that it can deal with very large datasets and is more robust than the $k$-means. Note also that by its recursive nature, another important feature is that it  allows automatic update and does not need to store all the data.
A key tuning parameter in our algorithm is the descent step value. We found empirically that reasonable values are given by the empirical $L_1$ loss function. We thus also consider an automatic two steps procedure in which one first runs  the sequential version of the $k$-means in order to approximate the value of the $L_1$ loss function and then run our stochastic $k$-medians with an appropriate descent step.

\medskip

The paper is organized as follows. We first fix notations and present our algorithm. In the third Section, we state the almost sure consistency of the stochastic gradient $k$-medians to a stationary point of the underlying objective function. The proof heavily relies on \cite{Monnez}. In Section 4, we compare on simulations the performance of our technique with the sequential $k$-means, the PAM algorithm and the trimmed $k$-means when the data are contaminated by a small fraction of outliers. We note that applying averaging techniques (see \cite{PolyakJud92})  to our estimator, with a small number of different  initializations points, is a very competitive approach even for moderate sample sizes with computation times that are much smaller.  In Section 5, we illustrate our new clustering algorithm on a large sample, of about 5000 individuals, in order to determine  profiles of television audience. A major difference with  PAM is that our algorithm searches for a solution in all the space whereas PAM, and its refinements CLARA and CLARANS, only look for a solution among the elements of the sample. Consequently, approaches such as PAM  are not adapted to deal with temporal data presented in Section 5 since the data mainly consist of 0 and 1 indicating that the television is switched on or switched off during each minute of the day. Proofs are gathered in the Appendix.

\section{The stochastic gradient $k$-medians algorithm}\label{sec:algo}
\subsection{Context and definitions}
Let $(\Omega, \mathcal A, \PP)$ be a probability space. Suppose we have a sequence of independent copies $Z_1, \ldots, Z_n$ of a random vector $Z$ taking values in $\Rset^d.$ The aim is to partition $\Omega$ into a finite number $k$ of clusters $\Omega_1, \ldots, \Omega_k$. Each cluster $\Omega_i$ is represented by its center, which is an element of $\Rset^{d}$ denoted by $\theta^i$. From a population point of view, the $k$-means and $k$-medians algorithms aim at finding local minima of  the function $g$ mapping $\Rset^{dk}$ to $\Rset$ and defined as follows, for $x=(x^1, \ldots, x^k)'$ with for all $i$, $x^{i}\in \Rset^{d}$,
\begin{eqnarray}
g(x) &\egaldef& \E \left( \min_{r=1, \ldots, k} \Phi( \nrm{Z - x^r}) \right), 
\label{def1:g}
\end{eqnarray}
where $\Phi$ is a real, positive, continuous and non-decreasing function and the norm $\nrm{.}$ in $\Rset^{d}$ takes account of the dimension $d$ of the data, for $z \in \Rset^{d}$, $\nrm{z}^2 = d^{-1} \sum_{j=1}^d z_j^2$. The particular case  $\Phi(u) = u^2,$ leads to the classical $k$-means algorithm, whereas $\phi(u)=|u|$ leads to the k-medians.

\medskip

Before presenting our new recursive algorithm, let us introduce now some notations and recall the recursive $k$-means algorithm developed by  \cite{MacQueen}.  Let us denote by $I_r$ the indicator function, 
\[I_r(z;x)=\prod_{j=1}^{k}\ind{\nrm{z-x^r}\leq \nrm{z-x^j}},\]
which is equal to one when $x^r$ is the nearest point to $z,$ among  the set of points $x^i,$ $i=1,\dots, k.$
The $k$-means recursive algorithm proposed by \cite{MacQueen} starts with $k$ arbitrary groups, each containing only  one point, $X_1^1, \ldots,X_1^k.$ Then, at each iteration, the cluster centers are updated as follows,
\begin{eqnarray}
\label{def:MacQueenAlgo}
X_{n+1}^r&=&X_n^r-a_n^rI_r(Z_n;X_n) \left(X_n^r-Z_n \right),
\end{eqnarray}
where for $n\geq 2$, $a_n^r = (1+n_r)^{-1}$ and  $n_r= 1+\sum_{\ell=1}^{n-1} I_r(Z_\ell;X_\ell)$ is just the number of elements allocated to cluster $r$ until iteration $n-1$. For $n=1$, let $a_1^r=\frac{1}{2}$. This also means that $X_{n+1}^r$ is simply the barycenter of the elements allocated to cluster $r$ until iteration $n,$
\[
X_{n+1}^r = \frac{1}{1 + \sum_{\ell=1}^{n} I_r(Z_\ell;X_\ell)} \left( X_1^r+ \sum_{\ell=1}^n I_r(Z_\ell;X_\ell) Z_\ell \right).
\]
The interesting point is that this recursive algorithm is very fast and can be seen as a Robbins-Monro  procedure.

\subsection{Stochastic gradient $k$-medians algorithms}

 Assuming $Z$ has an absolutely continuous distribution, we have 
 $$\PP(\nrm{Z-x^i}=\nrm{Z-x^j})=0, \quad \mbox{for any }i \neq j \mbox{ and }x^{i}\neq x^j.$$
Then, the  $k$-medians approach relies on looking for minima, that may be local, of the function $g$ which can also be written as follows, for any $x$ such that $x^{j}\neq x^{i}$ when $i\neq j$,
\begin{equation}
g(x)=\sum_{r=1}^k \E[I_r(Z;x)\nrm{Z-x^r}].
\label{def:fnobj}
\end{equation}
In order to get an explicit Robbins-Monro algorithm representation, it remains to exhibit the gradient of $g$. Let us write $g$ in  integral form. Denoting by   $f$ the density of the random variable $Z,$ we have,
\[g(x)=\sum_{r=1}^k \int_{\Rset^d\setminus \{x^r\}} I_r(z;x)\nrm{z-x^r} f(z) \ dz.\]
For $j=1, \ldots, d$, it can be checked easily that
\[\frac{\partial}{\partial x^{r}_j}\left(\nrm{z-x^r}\right)=\frac{x^{r}_j-z_j}{\nrm{z-x^r}},\]
and since
\[I_r(z;x)\frac{\abs{x^{r}_j-z_j}}{\nrm{z-x^r}}f(z) \leq f(z), \quad \mbox{for } z\neq x^r,\]
the partial derivatives satisfy,
\[
\frac{\partial g}{\partial x^{r}_j} (x) \ = \ \int_{\Rset^d\setminus \{x^r\} }I_r(z;x)\frac{x^{r}_j-z_j}{\nrm{z-x^r}}f(z)\ dz.
\] 
We define, for $x \in \Rset^{dk},$ 
\begin{equation}
\label{nabla}
\nabla_rg(x)\ \egaldef \ \E\left[I_r(Z;x)\frac{x^r-Z}{\nrm{x^r-Z}}\right].
\end{equation}

\medskip

We can now present our stochastic gradient $k$-medians algorithm. Given a set of  $k$ distinct initialization points in $\Rset^d,$ $X_1^1, \cdots, X_1^k,$ the set of $k$ cluster centers is updated at each iteration as follows.  
For $r=1, \ldots, k,$ and $n\geq 1,$
\begin{eqnarray}
\label{def:algosto1}
X_{n+1}^r&=&X_n^r-a_n^rI_r(Z_n;X_n)\frac{X_n^r-Z_n}{\nrm{X_n^r-Z_n}}\\
&=& X_n^r-a_n^r\nabla_r g(X_n)-a_n^rV_n^r, \nonumber
\end{eqnarray}
with $X_n = (X_n^1, \cdots, X_n^k),$ and 
\[
V_n^r\egaldef I_r(Z_n;X_n)\frac{X_n^r-Z_n}{\nrm{X_n^r-Z_n}}-\E\left[I_r(Z_n;X_n)\frac{X_n^r-Z_n}{\nrm{X_n^r-Z_n}}\Bigg|\CF_n\right],
\]
$\CF_n=\sigma(X_1, Z_1, \ldots, Z_{n-1}).$ The steps $a_n^r,$ also called gains, are supposed to be $\CF_n$-measurable. %The choice of the sequence $(a_n^r)_n$ will be discussed in section~\ref{ssec:stepsizes}. 
We denote  by $\nabla g(x)=\left(\nabla_1 g(x), \ldots, \nabla_k g(x) \right)'$ the gradient of $g$  and define  $V_n\egaldef (V_n^1, \ldots V_n^k)'$. Let $A_n$ be the diagonal matrix of size $dk \times dk,$
\[A_n=\left(\begin{array}{ccccccc}
a_n^1 &        &  & &&&\\
      & \ddots &  & &&&\\
      &        &a_n^1 &&&&\\
      &        &  & \ddots &&&\\ 
      &        &  & & a_n^k &&\\
      &        &  & & &\ddots &\\
     &        &  & & & &a_n^k\\
\end{array} \right),\]
each $a_n^r$  being repeated $d$ times. Then, the $k$-medians algorithm can be written in a matrix way,
\begin{equation}
\label{algo-vec}
X_{n+1}=X_n-A_n\nabla g(X_n)-A_nV_n,
\end{equation}
which is a classical stochastic gradient descent.

\subsection{Tuning the stochastic gradient $k$-medians and its averaged version}\label{ssec:stepsizes}
The behavior of algorithm (\ref{def:algosto1}) depends on  the sequence of steps $a_n^r,$ $r \in \{1, \ldots, k\}$ and the vector of initialization $X_1.$
These two sets of tuning parameters play distinct roles and we mainly focus on the choice of the step values, noting that, as for the $k$-means, the estimation results must be compared for different sets of initialization points in order to get a better estimation of the cluster centers. Assume we have a sample of $n$ realizations $Z_1, \ldots, Z_n$ of $Z$ and a set of initialization points of the algorithm, the  selected estimate of the cluster centers is the one minimizing the following empirical risk,
\begin{eqnarray}
\label{def:emprisk}
R(X_n) & =& \frac{1}{n}\sum_{i=1}^n\sum_{r=1}^k I_r(Z_i;X_n) \nrm{Z_i-X_n^r}
\end{eqnarray}

Let us denote  by $n_r = 1+ \sum_{\ell =1}^{n-1} I_r(Z_\ell;X_\ell)$ the number of updating steps for cluster $r,$ until iteration $n-1$, for $r \in \{1, \ldots, k\}.$ A classical form of the descent steps $a_n^r$ can be  given by 
\begin{eqnarray}
\label{def:anr}
a_n^r & = &  \left\{\displaystyle  \begin{array}{ll} a_{n-1}^r & \mbox{if } I_r(Z_n;X_n) = 0, \\  \displaystyle \frac{c_\gamma}{\left(1+c_\alpha n_r \right)^\alpha} & \mbox{ otherwise,} \end{array} \right. 
\end{eqnarray}
where $c_\gamma,$ $c_\alpha$ and $1/2 < \alpha \leq 1$ control the gain. 

Adopting an asymptotic point of view, one could believe that $\alpha$ should be set to $\alpha=1$ with suitable  constants $c_{\alpha}$ and $c_{\gamma}$, which are unknown in practice, in order to attain the optimal parametric rates of convergence of Robbins Monro algorithms (see \textit{e.g.} \cite{Duf97}, Th. 2.2.12).
Our experimental results on simulated data  have shown that the convergence of algorithm (\ref{def:algosto1}) with descent steps defined in (\ref{def:anr})  is then very sensitive to the values of the parameters $c_\gamma$ and $c_\alpha$ which have to be chosen very carefully. A simulation study performed in the particular case $k=1$  by \cite{CardCC10}  showed that the direct approach could lead to inaccurate  results and is nearly always less effective than the averaged algorithm presented below, even for well chosen descent step values. From an asymptotic point of view, it has been proved in  \cite{CCZ10} that the averaged stochastic gradient estimator of the geometric median, corresponding to $k=1,$ is asymptotically efficient under classical assumptions. Intuitively, when the algorithm is not too far from the solution, averaging allows to decrease substantially the variability of the initial algorithm which oscillates around the true solution and thus improves greatly its performances.

Consequently, we prefer to introduce an averaging step (see for instance \cite{PolyakJud92} or \cite{Pel00}), which  does not  slow down the algorithm and provides an estimator which is  much more effective. Our averaged estimator of the cluster centers, which remains recursive, is defined as follows, for $r \in \{1, \ldots, k\}$, $n\geq 1$, and for the value $X_{n+1}^r$ obtained by combining (\ref{def:algosto1}) and  (\ref{def:anr}), 
\begin{eqnarray}
\label{def:avekmed}
\bar{X}_{n+1}^r & = &  \left\{\displaystyle  \begin{array}{ll} \bar{X}_{n}^r & \mbox{if } I_r(Z_n;X_n) = 0, \\  \displaystyle \frac{n_r\bar{X}_{n}^r + X_{n+1}^r}{n_r+1} & \mbox{ otherwise,} \end{array} \right.  
\end{eqnarray}
with starting points $\bar{X}_1^r = X_1^r,$ $r=1, \ldots, k.$
Then standard choices (see \textit{e.g.} \cite{Bot10} and references therein) for $\alpha$ and $c_\alpha$ are  $\alpha=3/4$ and $c_\alpha=1,$ so that one only needs to select values for $c_\gamma.$

\section{Almost sure convergence of the algorithm}\label{sec:conv}
\subsection{A convergence theorem}

The following theorem is the main theoretical result of this paper. It states that the recursive algorithm defined in (\ref{algo-vec}) converges almost surely to the set of stationary points of the objective function defined in (\ref{def:fnobj}), under the following assumptions.
\begin{itemize}
\item[(H1)] a) The random vector $Z$ is absolutely continuous with respect to Lebesgue measure. \\
 b) $Z$ is bounded: $\exists K>0$: $\nrm{Z}\leq K$ a.s. \\
 c) $\exists C$: $\forall x \in \Rset^{d}$ such that $\nrm{x}\leq K+1$, $\E\left[\frac{1}{\nrm{Z-x}}\right]<C$. 
\item[(H2)] 
a) $\forall n \geq 1$,  $\min_ra_n^r>0$. \\
b) $\max_r\sup_n a_n^r < \min(\frac12, \frac{1}{8C})$ a.s. \\
c) $\sum_{n=1}^{\infty}\max_r a_n^r=\infty$ a.s. \\
d) $ \sup_{n}\frac{\max_ra_n^r}{\min_r a_n^r}<\infty \quad \mbox{a.s.}$
\item[(H3)] $\sum_{r=1}^k\sum_{n=1}^{\infty}\left(a_n^r\right)^2<\infty \quad \mbox{a.s.}$

\item[(H3')] $\sum_{r=1}^k\sum_{n=1}^{\infty}\E\left[\left(a_n^r\right)^2I_r(Z_n; X_n)\right]<\infty .$
\end{itemize}

\begin{thm}\label{thm-principal}
Assume that $X_1$ is absolutely continuous and that $\nrm{X_1^r}\leq K$, for $r=1, \ldots, k$. Then under Assumptions (H1a,c), (H2a,b), (H3) or (H3'), $g(X_n)$ and 
\[\sum_{r=1}^k\sum_{n=1}^{\infty}a_n^r\nrm{\nabla_r g(X_n)}^2\]
 converge almost surely. \\
 Moreover, if the hypotheses  (H1b) and (H2c,d) are also fulfilled then $\nabla g(X_n)$ and the distance between $X_n$ and  the set of stationary points of $g$ converge almost surely to zero.
\end{thm}

A direct consequence of Theorem 1 is that if the set of stationary points  of $g$ is finite,  then the sequence $(X_n)_n$ necessarily converges almost surely towards one of these stationary points because $X_{n+1}-X_n$ converges almost surely towards zero. By Cesaro means arguments, the averaged sequence $\bar{X}_n$  also converges almost surely towards the same stationary point.

\subsection{Comments on the hypotheses}
Note first that if the data do not arrive online and $X_1$  is chosen randomly among the sample units then $X_1$ is absolutely continuous and $\nrm{X_1^r}\leq K$, for $r=1, \ldots, k$ under (H1a,b). Moreover, 
the absolute continuity of $Z$ is a technical assumption that is required to get decomposition (\ref{def:fnobj}) of the $L_1$ error. 
Proving the convergence in the presence of atoms would require to decompose this error in order to take into account the  points which could have a non-null probability to be at the same distance. Such a study is clearly beyond the scope of the paper. Note however that in the simple case $k=1,$ 
it has been established in \cite{CCZ10} that the stochastic algorithm for the functional median is convergent provided that the distribution, which can be a mixture of a continuous and a discrete distribution, does not charge the median.

Hypothesis (H1c) is a stronger version of a more classical hypothesis needed to get consistent estimators of the spatial median (see \cite{Cha96}). As noted in \cite{CCZ10}, it is closely related to small ball properties of $Z$ and is fulfilled when
$$
\PP \left( \nrm{Z-x} \leq \epsilon \right) \leq \kappa \epsilon^2,
$$
for a constant $\kappa$ that does not depend on $x$ and $\epsilon$ small enough.
This implies in particular that  hypothesis (H1c) can be satisfied only when the dimension $d$ of the data satisfies $d \geq 2$.

Hypotheses (H2) and (H3) or (H3') deal with the stepsizes. Considering the general form of gains $a_n^r$ given in  (\ref{def:anr}), they are fulfilled when the sizes $n_r$ of all the clusters grow to infinity at the same rate and $\alpha \in ]1/2, 1]$.

%%%%%%%%%%%%%%%%%%%%%%%%%%%%%%%%%%%%%%%%%%%
\section{A simulation study}\label{sec:simu}

We first perform a simulation study to compare our recursive $k$-medians algorithm with the following well known clustering algorithms :  recursive version of the $k$-means (function \texttt{kmeans} in \Rlogo), trimmed $k$-means  (function \texttt{tkmeans}  in the \Rlogo \ package \texttt{tclust}, with a trimming coefficient $\alpha$ set to default, $\alpha=0.05$) and PAM (function \texttt{pam} in the \Rlogo \ package \texttt{cluster}). Our \Rlogo \ codes are available on request.

Comparisons are first made  according to the  value of the empirical $L_1$ error  (\ref{def:emprisk}) which must be as small as possible.
We note that the results of our averaged recursive procedure defined by (\ref{def:algosto1}),  (\ref{def:anr}) and (\ref{def:avekmed}) are very stable when the value of the tuning parameter $c_\gamma$ is not too far from the minimum value of the $L_1$ error, with $\alpha=3/4$ and $c_\alpha=1.$
This leads us to propose, in Section~\ref{sec:sensitivity},  an automatic clustering algorithm  which consists in first approximating the $L_1$ error with a recursive $k$-means and then performing our recursive $k$-medians with the selected value of $c_\gamma,$ denoted by $c$ in the following.
We have no mathematical justification for such an automatic choice of the tuning parameter $c$ but it always worked well on all our simulated experiments. This important point of our algorithm deserves further investigations that are beyond the scope of the paper.
Note however that this intuitive approach will certainly be ineffective when the dispersion is very different from one group to another.
It would then  be possible to consider refinements of the previous algorithm which would consist in considering different values of tuning parameter $c$  for the different clusters. 
We only present here a few simulation experiments which highlight both the strengths and the drawbacks of our recursive $k$-medians algorithm.

\subsection{Simulation protocol}

\subsection*{Simulation 1 : a simple experiment in $\mathbb{R}^2$}
We first consider a very simple case and draw independent realizations of variable $Z,$ 
\begin{eqnarray}
Z &=& (1-\epsilon) \left(\pi_1 Z_1 + \pi_2 Z_2+ \pi_3 Z_3\right) + \epsilon \delta_z,
\label{def:melange}
\end{eqnarray}
which is a mixture, with weights $\pi_1 = \pi_2 = \pi_3 = 1/3,$ of three bivariate random Gaussian vectors $Z_1$, $Z_2$ and $Z_3$ with mean vectors $\mu_1=(-3,-3),$  $\mu_2=(3,-3)$ and $\mu_3=(4.5,-4.5)$ and covariance matrices $Var(Z_1)= \left( \begin{array}{cc} 2 & 1 \\ 1 & 3 \end{array} \right),$   $Var(Z_2)= \left( \begin{array}{cc} 3 & 1 \\ 1 & 2 \end{array} \right)$ and  $Var(Z_3)= \left( \begin{array}{cc} 2 & -1 \\ -1 & 3 \end{array} \right).$ \\
Point $z =(-14,14)$ is an outlier and parameter $\epsilon$ controls the level of the contamination. 
 A sample of $n=450$ realizations of $Z$ is drawn in Figure~\ref{fig1}. 
 
\begin{figure}[htb]
\begin{center}
\includegraphics[height=11cm,width=11cm]{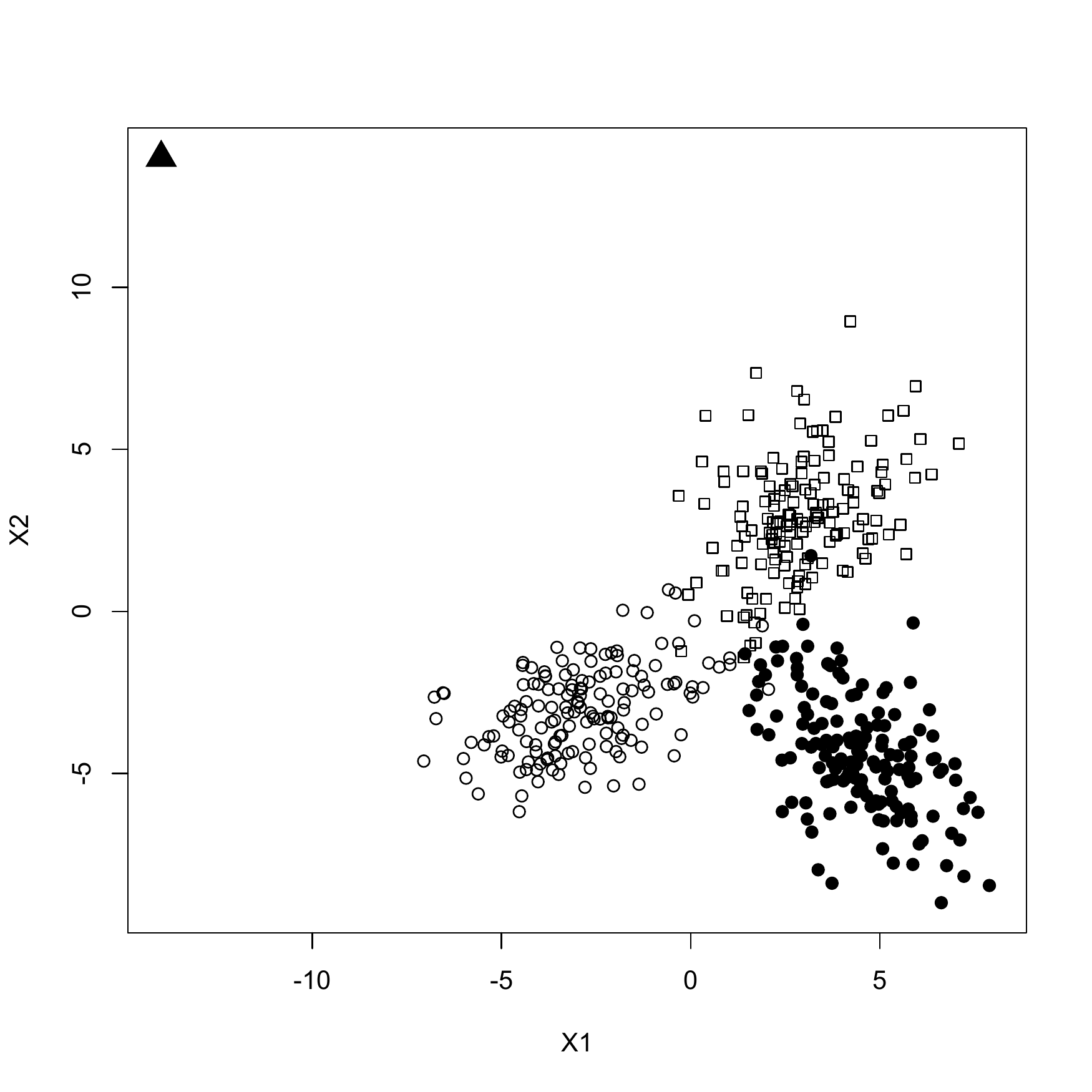} 
\caption{Simulation 1. A sample of $n=450$ realizations of $Z$. An outlier is located at position (-14,14).}
\label{fig1}
\end{center}
\end{figure}

\subsection*{Simulation 2 : larger dimension with different correlation levels}
\begin{figure}[htb]
\begin{center}
\includegraphics[height=11cm,width=11cm]{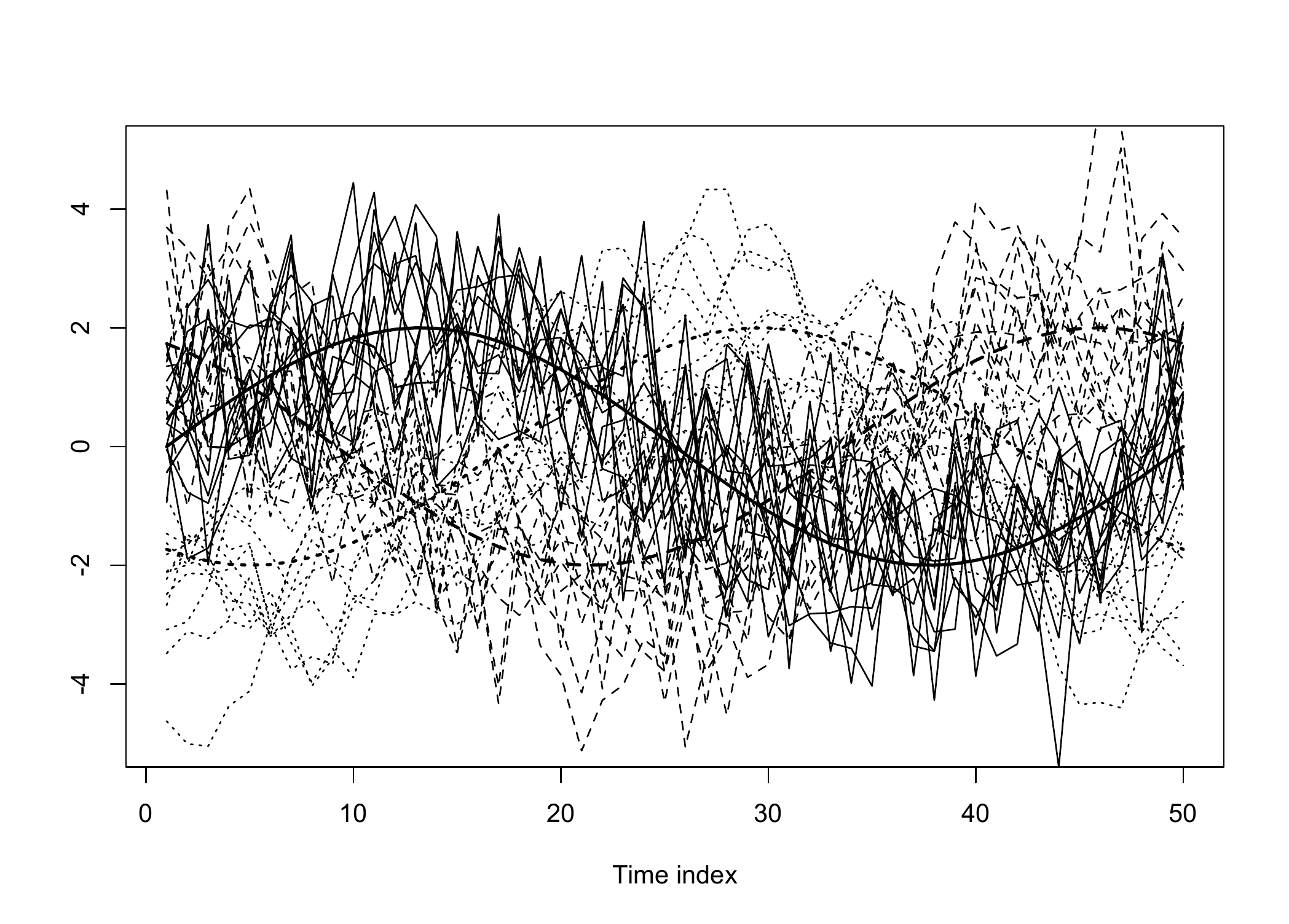} 
\caption{Simulation 2. A sample of $n=36$ realizations of $Z$ with $d=50$. The mean values $\mu_1,$ $\mu_2$ and $\mu_3$ of the three natural  clusters are drawn in bold lines.}
\label{figsim201}
\end{center}
\end{figure}

We also performed a simulation experiment, with a mixture of three Gaussian random variables as in  (\ref{def:melange}),  but in higher dimension spaces with correlation levels that vary from one cluster to another. Now, $Z_1,$ $Z_2$  and $Z_3$ are independent  multivariate normal distributions in $\mathbb{R}^{d},$ with means $\mu_{1j} = 2 \sin(2 \pi j/(d-1)),$  $\mu_{2j} = 2 \sin(2\pi/3 + 2 \pi j/(d-1)),$ and $\mu_{3j} = 2 \sin(4\pi/3 + 2 \pi j/(d-1)),$  for $j=1, \ldots, d.$ The covariance functions $Cov(Z_{ij},Z_{i\ell}) = 1.5 \rho_i^{|j- \ell|},$ for $j, \ell \in {1, \ldots, d}$ and $i \in \{1,3\}$ are controlled by a correlation parameter $\rho,$ with $\rho_1=0.1$, $\rho_2 = 0.5$ and $\rho_3=0.9.$ 
Note that this covariance structure corresponds to autoregressive processes of order one with autocorrelation $\rho.$
As before, $\delta_z = (4,\ldots, 4) \in \mathbb{R}^{d}$ plays the role of an outlying point.
A sample of $n=36$ independent realizations of $Z,$ without outliers, is  drawn in Figure \ref{figsim201} for a dimension $d=50.$

\subsection{$L_1$ error  and sensitivity to parameter $c$}
\label{sec:sensitivity}
As noted in  \cite{BryantWilliamson78}, comparing directly the distance of the estimates from the cluster centers $\mu_1,$ $\mu_2$ and $\mu_3$ may not be appropriate to evaluate a clustering method. 
Our comparison is thus first made in terms of the value of the empirical $L_1$ error (\ref{def:emprisk}) which should be as small as possible. 
For all methods, we considered that there were $k=3$ clusters.

\begin{figure}[htb]
\begin{center}
\includegraphics[height=11cm,width=11cm]{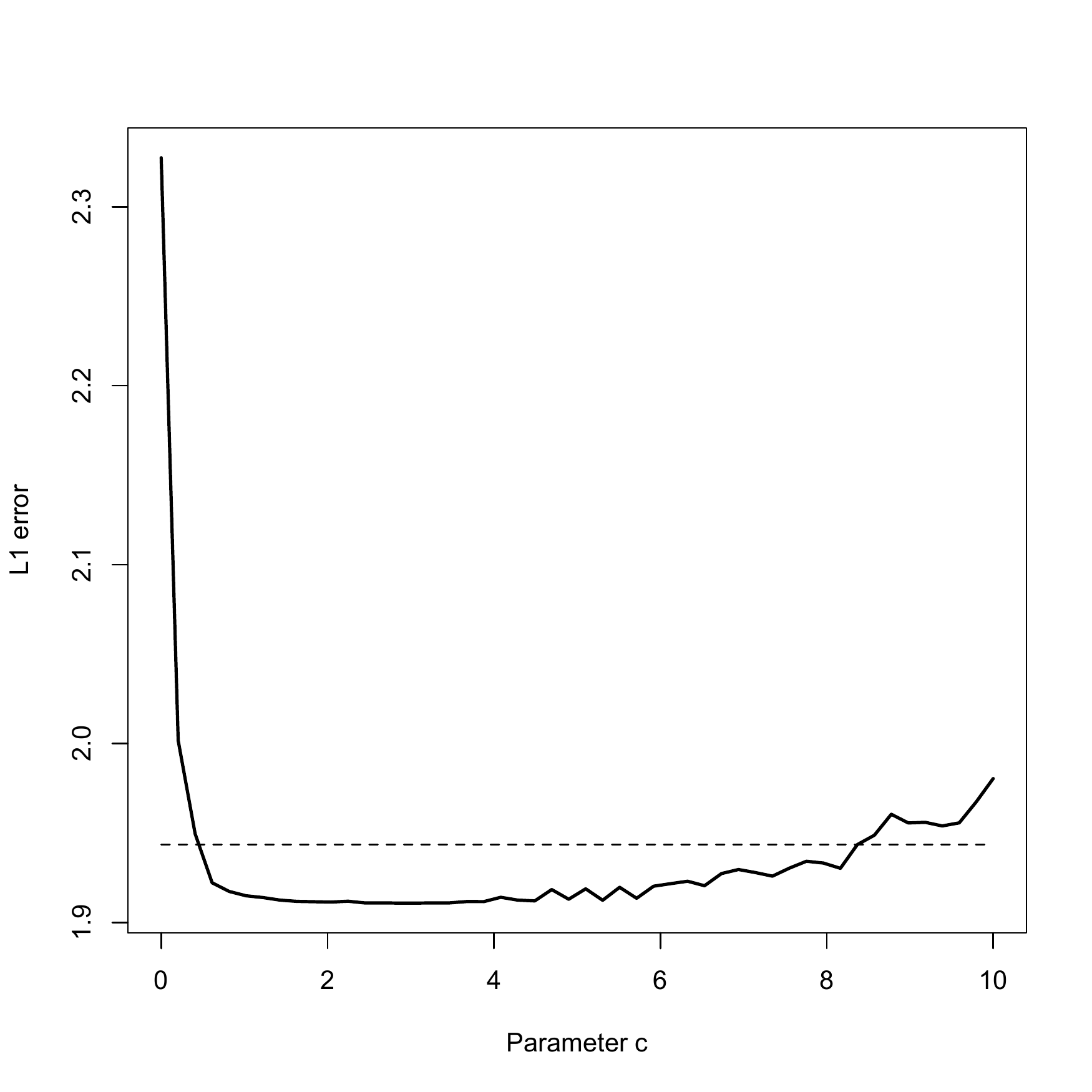} 
\caption{Simulation 1 with $\epsilon=0.05$ and $n=250.$ Mean empirical $L_1$ error (over 50 replications) for the PAM algorithm (dashed line), the $k$-means ($c=0$) and the stochastic $k$-medians (bold line), for $c \in ]0,10].$}
\label{figsim1L1}
\end{center}
\end{figure}

\medskip

We first study the simple case of Simulation 1. The empirical mean $L_1$ error of the PAM algorithm, the $k$-means and the averaged $k$-medians, for  50 replications of samples with sizes $n=250$ and a contamination level $\epsilon=0.05$ is presented in Figure~\ref{figsim1L1}. The number of initialization points equals 10 for both the $k$-means and the $k$-medians. When the descent parameter $c$ equals 0, the initialization 
point is given by the estimated centers by the $k$-means, so that the empirical $L_1$ error corresponds in that case to the $k$-means error, which is sightly above 2.31. We first note that this $L_1$ error is always larger, even if the contamination level is small, than the PAM and the $k$-medians errors, for $c \in ]0,10].$ Secondly, it appears that for $c \in [0.5, 4],$ the $k$-medians $L_1$ error is nearly constant and is clearly smaller than the $L_1$ error of the PAM algorithm. This means that, even if the sample size is relatively small ($n=250$), the recursive $k$-medians can perform well for values of $c$ which are of the same order of the $L_1$ error.

\medskip
 
 \begin{figure}[htbp]
\begin{center}
\includegraphics[height=11cm,width=11cm]{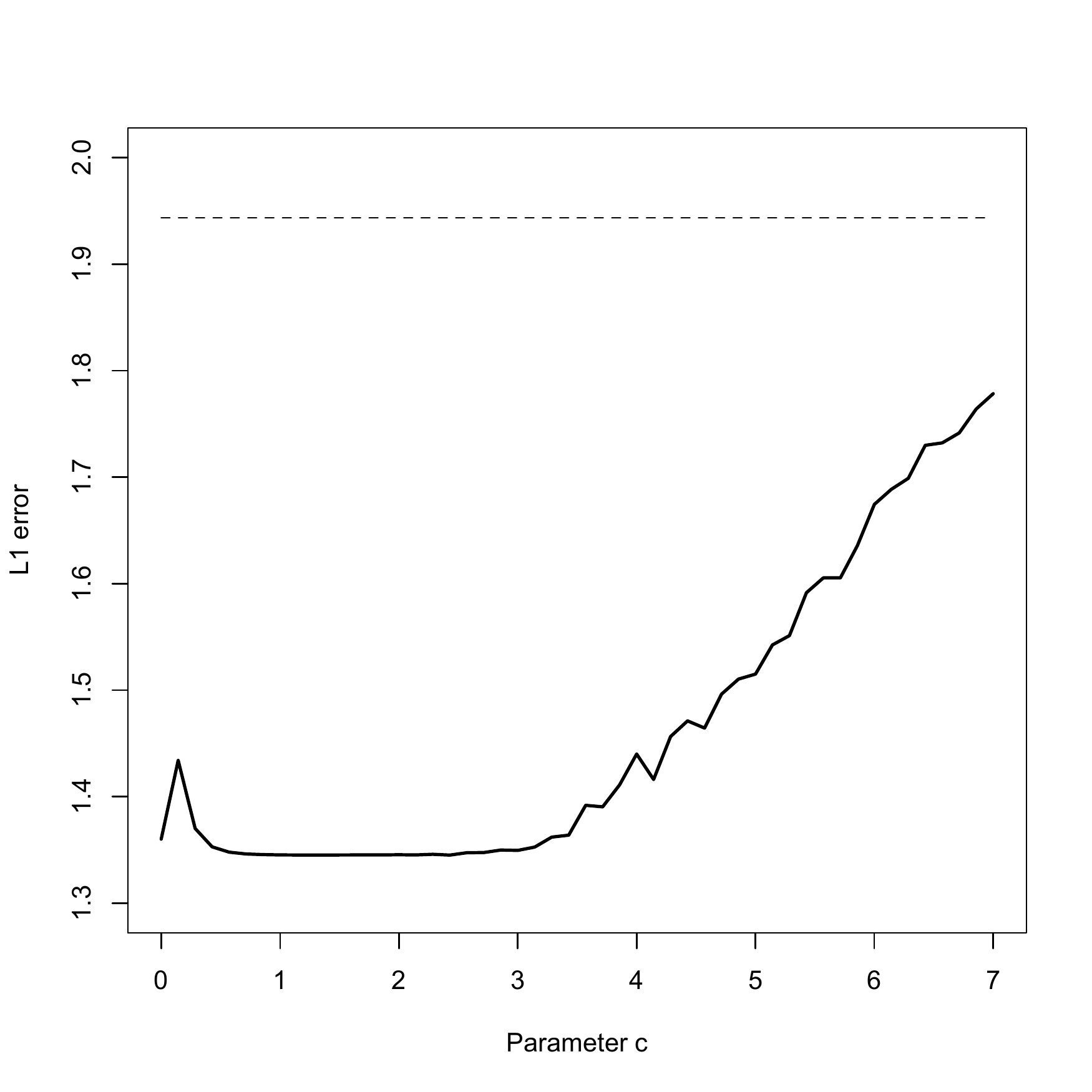} 
\caption{Simulation 2 with $n=500,$ $d=50,$ and $\epsilon=0.05.$ The mean empirical $L_1$ error (over 50 replications) is represented for the PAM algorithm (dashed line), the MacQueen version of the $k$-means ($c=0$) and the recursive $k$-medians estimator (bold line), for $c \in ]0,7].$}
\label{figsim2L1p50}
\end{center}
\end{figure}

We now consider 50 replications of samples drawn from the distribution described in simulation 2, with $n=500,$ $d=50$ and $\epsilon=0.05.$
The number of initialization points for the $k$-means and the $k$-medians is now equal to 25 and the empirical mean $L_1$ error is presented in Figure~\ref{figsim2L1p50}. We first note that the performances of the PAM algorithm clearly decrease with the dimension. The $k$-means performs better even if there are 5\% of outliers and if it is not designed to minimize an $L_1$ error criterion. This can be explained by the fact that PAM, as well as CLARA and CLARANS, look for a solution among the elements of the sample. Thus these approaches can hardly explore all the dimensions of the data when $d$ is large and $n$ is not large enough.   On the other hand, the $k$-medians and the $k$-means look for a solution in all $\mathbb{R}^d$ and are not restricted to the observed data and thus provide better results in terms of $L_1$ error.
As before, we can also remark  that the minimum error, which is around 1.36, is attained for $c$ in the interval $[0.5,3].$

\begin{figure}[htbp]
\begin{center}
\includegraphics[height=12cm,width=12cm]{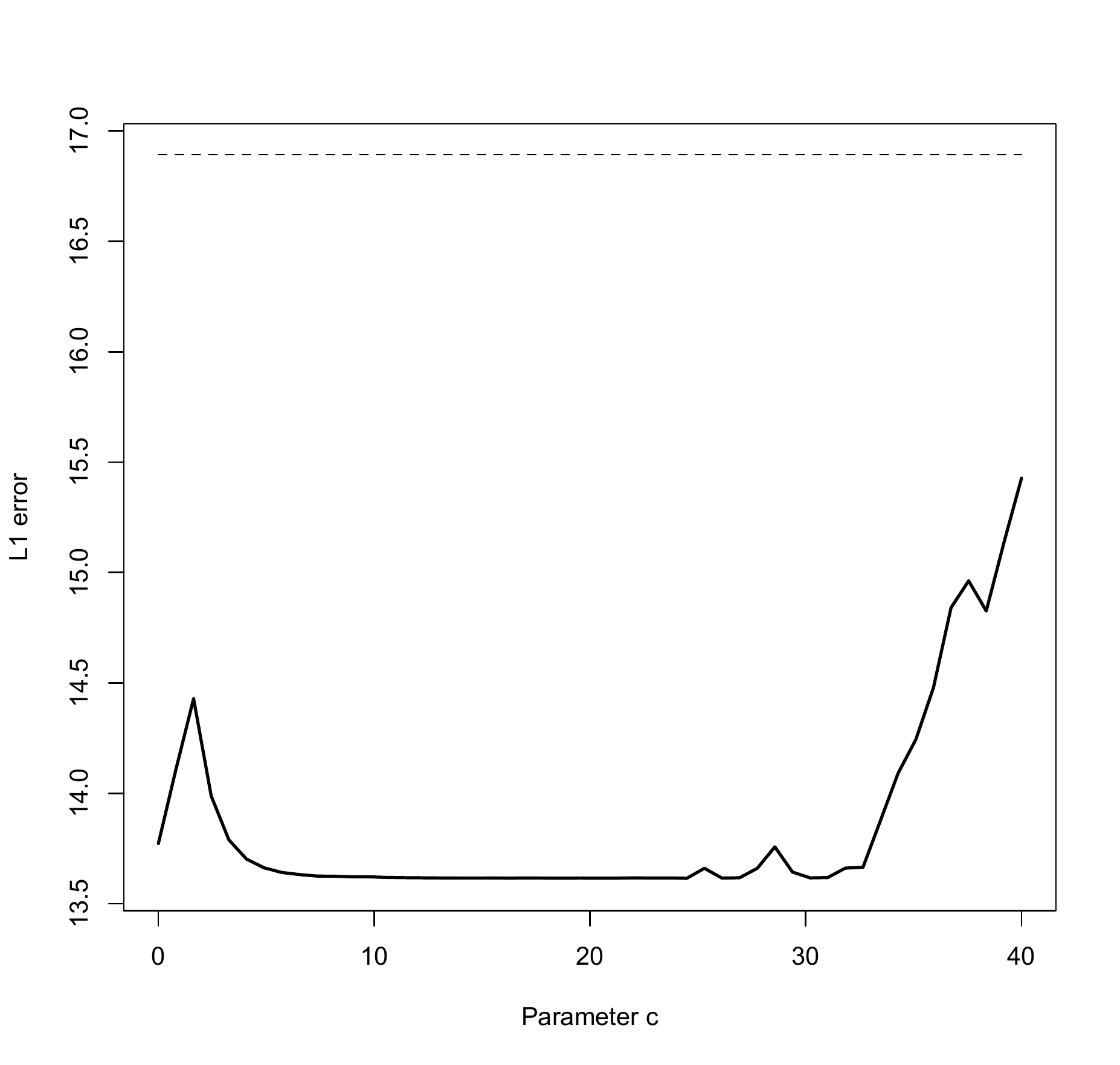} 
\caption{Simulation 2 with $n=1000,$  $d=200,$ $\epsilon=0.05,$ and $Z$ multiplied by a factor 10. The mean $L_1$ loss function (over 50 replications) is represented for the PAM algorithm (dashed line), the MacQueen version of the $k$-means ($c=0$) and our recursive $k$-medians estimator (bold line), for $c \in ]0,40].$}
\label{figsim2L1p200}
\end{center}
\end{figure}

\medskip

We finally present results from Simulation 2 in which we consider samples with size $n=1000,$ of variable $10 Z,$ with $d=200.$ The contamination level is $\epsilon=0.05$ and 50 initialization points were considered for the $k$-means and $k$-medians algorithms. Since $Z$ has been multiplied by a factor 10, the minimum of the $L_1$ error is now around 13.6. We remark, as before, that because of the dimension of the data, $d=200,$ PAM is outperformed by the $k$-means and the $k$-medians even in the presence of a small fraction of outliers ($\epsilon=0.05$). The minimum of the $L_1$ error for the $k$-medians estimator is again very stable for $c \in [5,25]$ with smaller values than the $L_1$ error of the $k$-means clustering.

\medskip

As a first conclusion, it appears that for large dimensions the $k$-medians can give results which are much better than PAM in terms of empirical $L_1$ error. We can also note that the averaged recursive $k$-medians is not very sensitive to the choice of parameter $c$ provided its value is not too far from the minimum value of the $L_1$ error. Thus we only consider, in the following subsection, the data-driven version of our averaged algorithm described in Section~\ref{ssec:stepsizes} in which the value of $c$ is chosen automatically, its value being  the empirical  $L_1$ error of the  recursive $k$-means.
This data-driven $k$-medians algorithm can be summarized as follows
{\em \begin{enumerate}
\item Run the $k$-means algorithm  and get the estimated centers.
 \item Set $c$ as the value of the $L_1$ error of the $k$-means, evaluated  with formula~(\ref{def:emprisk}).
 \item Run the averaged $k$-medians defined by (\ref{def:algosto1}),  (\ref{def:anr}) and (\ref{def:avekmed}), with $c$ computed in Step 2 and $c_\alpha=1.$
\end{enumerate} 
}

\subsection{Classification Error Rate}

We now make comparisons in terms of classification error measured by the Classification Error Rate (CER) introduced by \cite{ChipTib2005} and defined as follows. For a given partition $P$ of the sample,  let $1_{P(i,i')}$ be an indicator for whether partition $P$ places observations $i$ and $i'$ in the same group. Consider a partition $Q$ with the true class labels, the CER for partition $P$ is defined as 
\begin{eqnarray}
\mbox{CER} &=&  \frac{2}{n(n-1)} \sum_{i > i'} \left| 1_{P(i,i')} - 1_{Q(i,i')} \right| .
\label{def:CER}
\end{eqnarray}
The CER equals 0 if the partitions $P$ and $Q$ agree perfectly whereas a high value indicates disagreement.  Since PAM, the  $k$-means and  our algorithm are not designed to detect  outliers automatically, we only evaluate the CER  on the non-outlying pairs of elements $i$ and $i'$ of the sample.

\begin{figure}[htbp]
\begin{center}
\includegraphics[height=10cm,width=14cm]{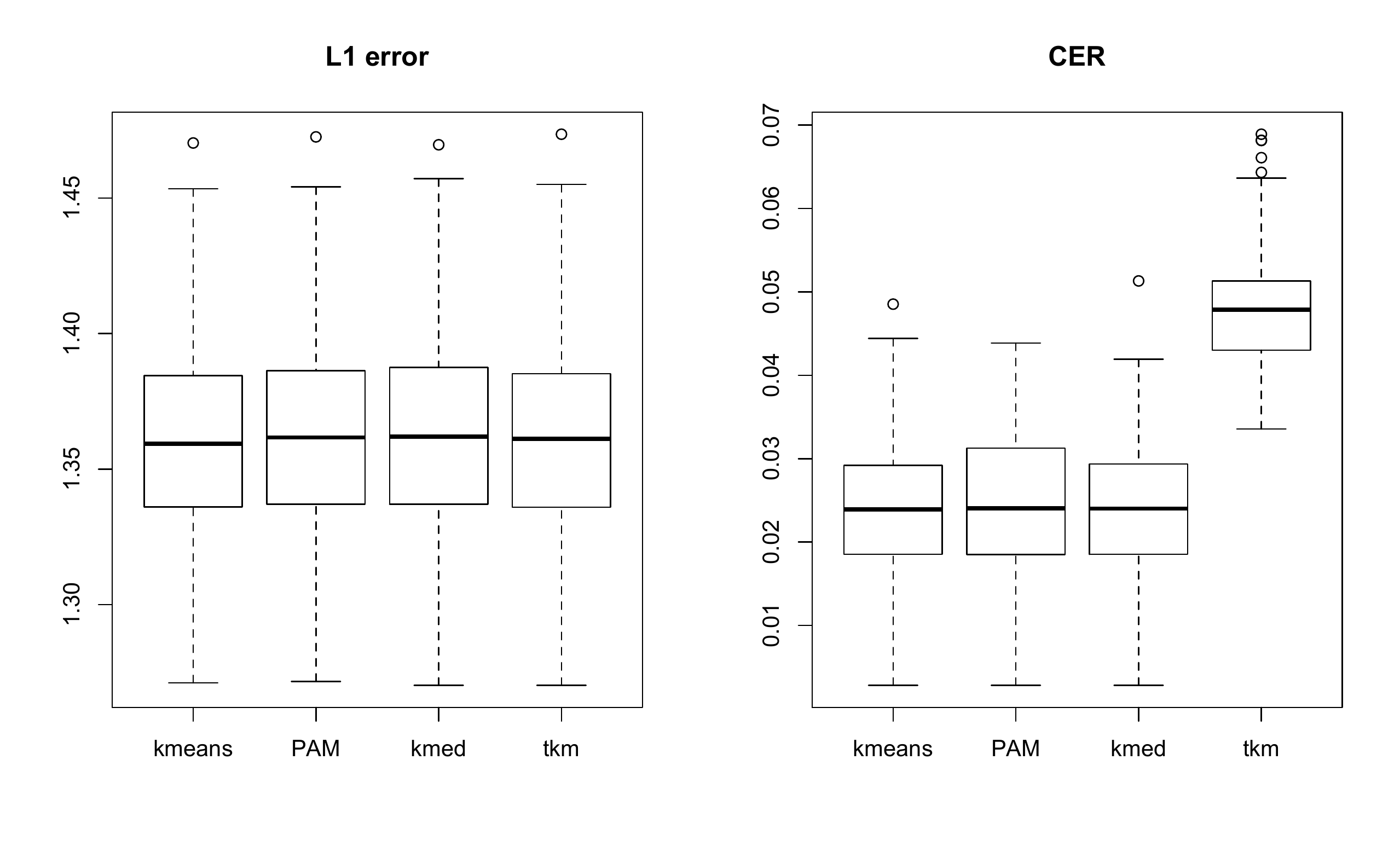} 
\caption{Simulation 1 with $\epsilon=0$ and $n=500.$ On the left, the $L_1$ empirical error. On the right, CER for the $k$-means, PAM, the data-driven recursive $k$-medians  algorithm (kmed) and the trimmed $k$-means (tkm).}
\label{figsim1L1n500}
\end{center}
\end{figure}

We present in Figure~\ref{figsim1L1n500}, results for 500 replications of Simulation 1, with a sample size $n=500$ and no outliers ($\epsilon=0$). We note, in this small dimension context with no contamination, that the $L_1$ errors  are comparable.
Nevertheless, in terms of CER, the PAM, the $k$-means and  the data-driven $k$-medians algorithms have approximately the same performances. For the trimmed $k$-means, the results are not as effective, since this algorithm automatically classifies 5\% of the elements of the sample as outliers.

\begin{figure}[htbp]
\begin{center}
\includegraphics[height=10cm,width=14cm]{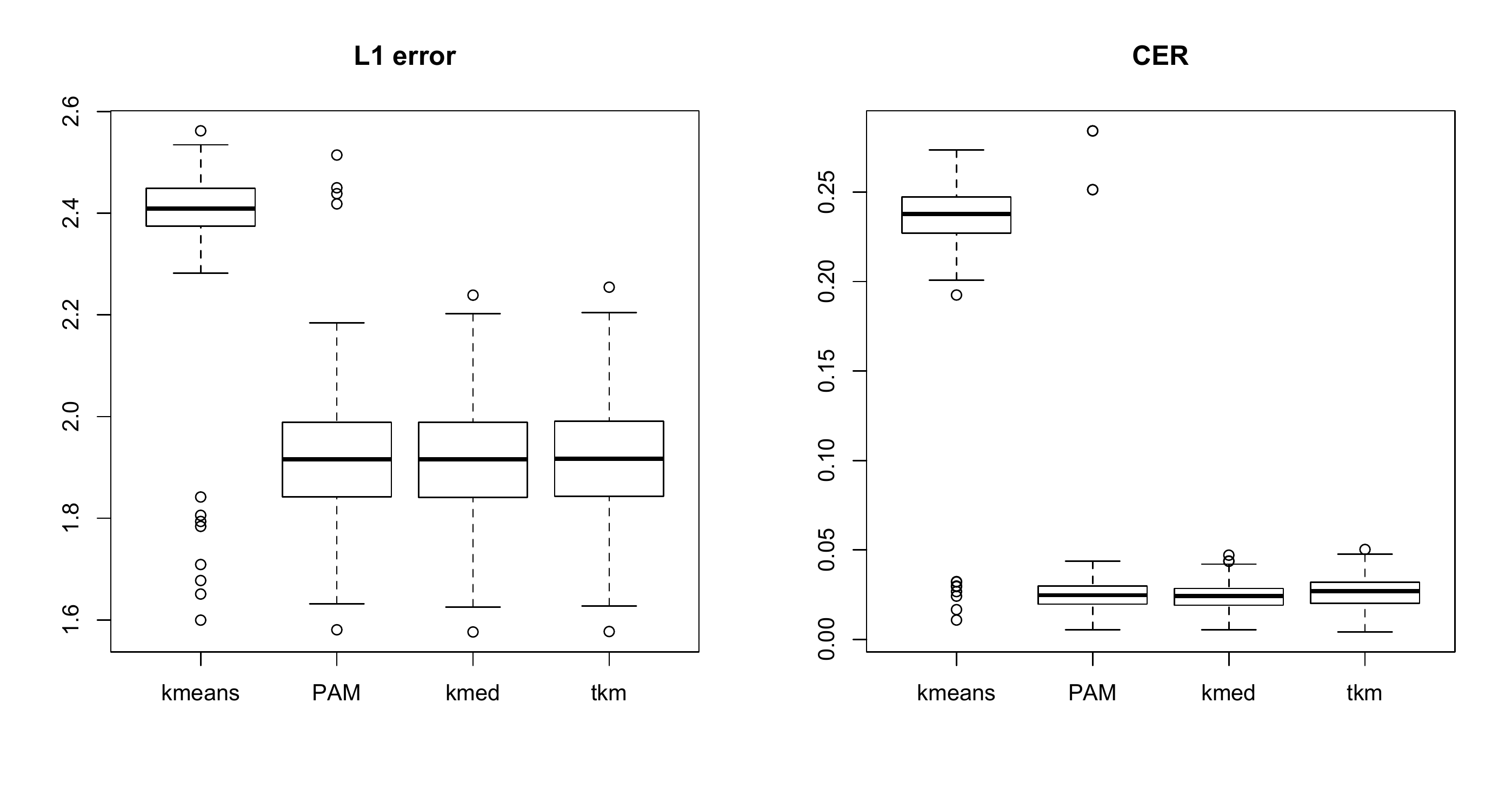} 
\caption{Simulation 1 with $\epsilon=0.05$ and $n=500.$ On the left, the $L_1$ empirical error. On the right, CER for the $k$-means, PAM, the data-driven recursive $k$-medians  algorithm (kmed) and the trimmed $k$-means (tkm).}
\label{figsim1L1n500d05}
\end{center}
\end{figure}

We then consider the same experiment as before, the only difference being that there is now a fraction of $\epsilon=0.05$ of outliers. The results are presented in Figure~\ref{figsim1L1n500d05}. The $k$-means algorithm is clearly affected by the presence of outliers and both its $L_1$ error and its CER are now much larger than for the other algorithms. PAM and the recursive $k$-medians have similar performances, even if PAM is slightly better. The trimmed $k$-means now detects the outliers  and also has  good  performances.
If the contamination level increases to $\epsilon=0.1$, as presented in Figure~\ref{figsim1L1n1000d1}, then PAM and the trimmed $k$-means (with a trimming coefficient $\alpha=0.05$) are strongly affected in terms of CER and do not recover the true groups. The $k$-medians algorithm is less affected by this larger level of contamination. Its median CER is nearly unchanged, meaning that for at least 50 \% of the samples, it gives a correct partition.

\begin{figure}[htbp]
\begin{center}
\includegraphics[height=10cm,width=14cm]{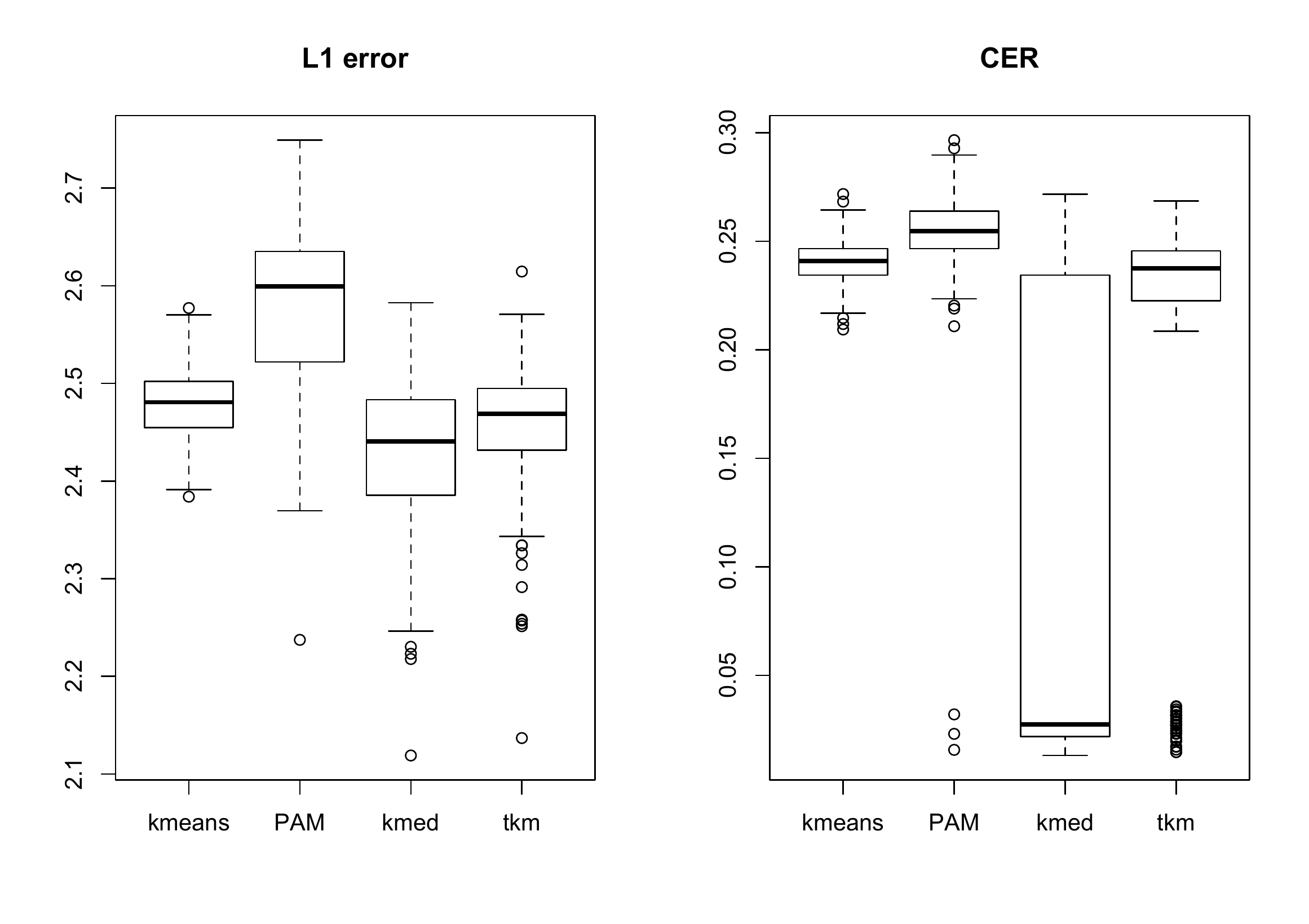} 
\caption{Simulation 1 with $\epsilon=0.1$ and $n=1000.$ On the left, the $L_1$ empirical error. On the right, CER for the $k$-means, PAM, the data-driven recursive $k$-medians  algorithm (kmed) and the trimmed $k$-means (tkm).}
\label{figsim1L1n1000d1}
\end{center}
\end{figure}

We now consider Simulation 2, with a dimension $d=50$ and a fraction $\epsilon=0.05$ of outliers.  The $L_1$ empirical errors and the CER, for sample sizes $n=500,$ are given in Figure~\ref{figsim2L1d05n500}. It clearly appears that PAM has the largest $L_1$ errors and the trimmed $k$-means and the data-driven $k$-medians the smallest ones. Intermediate $L_1$ errors are obtained for the $k$-means. In terms of CER, the partitions obtained by the $k$-means and PAM are not effective and do not recover well the true partition in the majority of the samples. The trimmed $k$-means and our algorithm always perform well and have similar low values in terms of  CER.

\begin{figure}[htbp]
\begin{center}
\includegraphics[height=10cm,width=14cm]{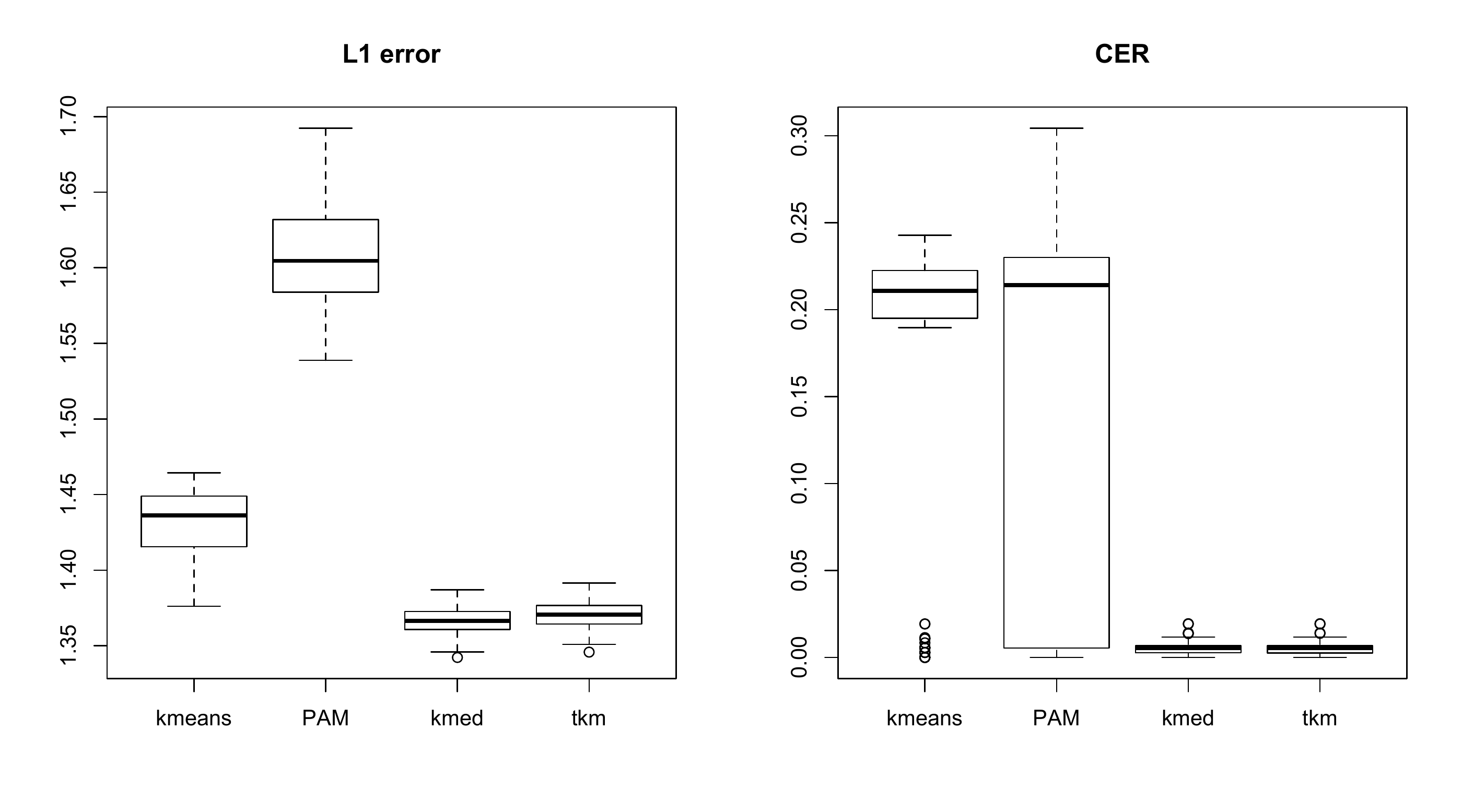} 
\caption{Simulation 2 with $\epsilon=0.05,$ $n=500$ and $d=50.$ On the left, the $L_1$ empirical error. On the right, CER for the $k$-means, PAM, the data-driven recursive $k$-medians  algorithm (kmed) and the trimmed $k$-means (tkm).}
\label{figsim2L1d05n500}
\end{center}
\end{figure}

\subsection{Computation time}

The \Rlogo \ codes of all the considered estimation procedures call \texttt{C} routines and provide the same output. Mean computation times, for 100 runs, various sample sizes and numbers of clusters are reported in Table \ref{table:temps}. They are based on one initialization point. From a computational point of view, the recursive $k$-means based on the MacQueen algorithm as well as the  averaged stochastic $k$-medians algorithm  are always faster than the other algorithms and the gain increases as the sample size gets larger. For example, when $k=5$ and $n=2000,$ the stochastic $k$-medians  is approximately 30 times faster than the trimmed $k$-means and 350 times faster than the PAM algorithm. The data-driven recursive $k$-medians has a computation time of approximately the sum of the computation time of the recursive $k$-means and the stochastic $k$-medians. This also means that when the allocated computation time is fixed and the dataset is very large, the data-driven recursive $k$-medians can deal with sample sizes that are 15 times larger than the trimmed $k$-means and 175 times larger than the PAM algorithm. 

\begin{table}[htdp]
\caption{Comparison of the mean computation time in seconds, for 100 runs,  of the different estimators for various sample sizes and number of clusters $k.$ The computation time are given for one initialization point. }
{\small \begin{center}
\begin{tabular}{|c|ccc|ccc|ccc|} \hline
 &  & n=250 &   &  & n=500 &  &   & n=2000 &  \\ 
Estimator & k=2 & k=4 & k=5 & k=2 & k=4 & k=5 & k=2 & k=4 & k=5 \\ \hline
 $k$-medians & 0.33 & 0.35 & 0.36  & 0.45 & 0.47 & 0.48  & 1.14 & 1.25 & 1.68 \\ \hline
PAM & 1.38 & 3.34 & 4.21 & 5.46  & 15.12  & 20.90 & 111 & 302.00 &  596.00 \\ \hline
Trimmed $k$-means &  2.20 & 5.45 & 6.76  & 5.32 &11.19 &  13.48 & 22.97  & 42.72 &   51.00 \\ \hline
 MacQueen & 0.21 & 0.29 & 0.31 & 0.25 & 0.43 & 0.50  & 0.60 & 1.38 &  1.76 \\ \hline
\end{tabular}
\end{center}
}
\label{table:temps}
\end{table}

When the sample size or the dimension increases, the computation time is even more critical. For instance, when $d=1440$ and $n=5422$ as in the example of Section~\ref{sec:Mediametrie},  our data-driven recursive $k$-medians estimation procedure is  at least 500 times faster than the trimmed $k$-means.  It takes 5.5 seconds for the sequential $k$-means to converge and then about 3.0 seconds for the averaged $k$-medians, whereas  it takes more than 5700 seconds for the trimmed $k$-means.

%%% Mediametrie
\section{Determining television audience profiles with $k$-medians}\label{sec:Mediametrie}

The M\'ediam\'etrie company provides every day the official estimations of television audience in France. Television consumption can be measured both in terms of how long  people watch each channel and when they watch television. M\'ediam\'etrie  has  a panel of about 9000 individuals equipped at home with sensors that are able to record  and send the audience of the different television channels. Among this panel, a sample of around 7000 people is drawn every day and the television consumption  of the people belonging to this sample is sent to M\'ediam\'etrie at night, between 3 and 5 am. Online clustering techniques are then interesting to determine automatically, the number of clusters being fixed in advance,  the main profiles of viewers and then  relate these profiles to socio-economic variables.  
In these samples, M\'ediam\'etrie  has noted the presence of  some atypical behaviors so that robust techniques may be helpful.

%\subsection{Clustering temporal profiles}\label{mediam:minute}

\begin{figure}[htb]
\begin{center}
\includegraphics[height=10cm,width=14cm]{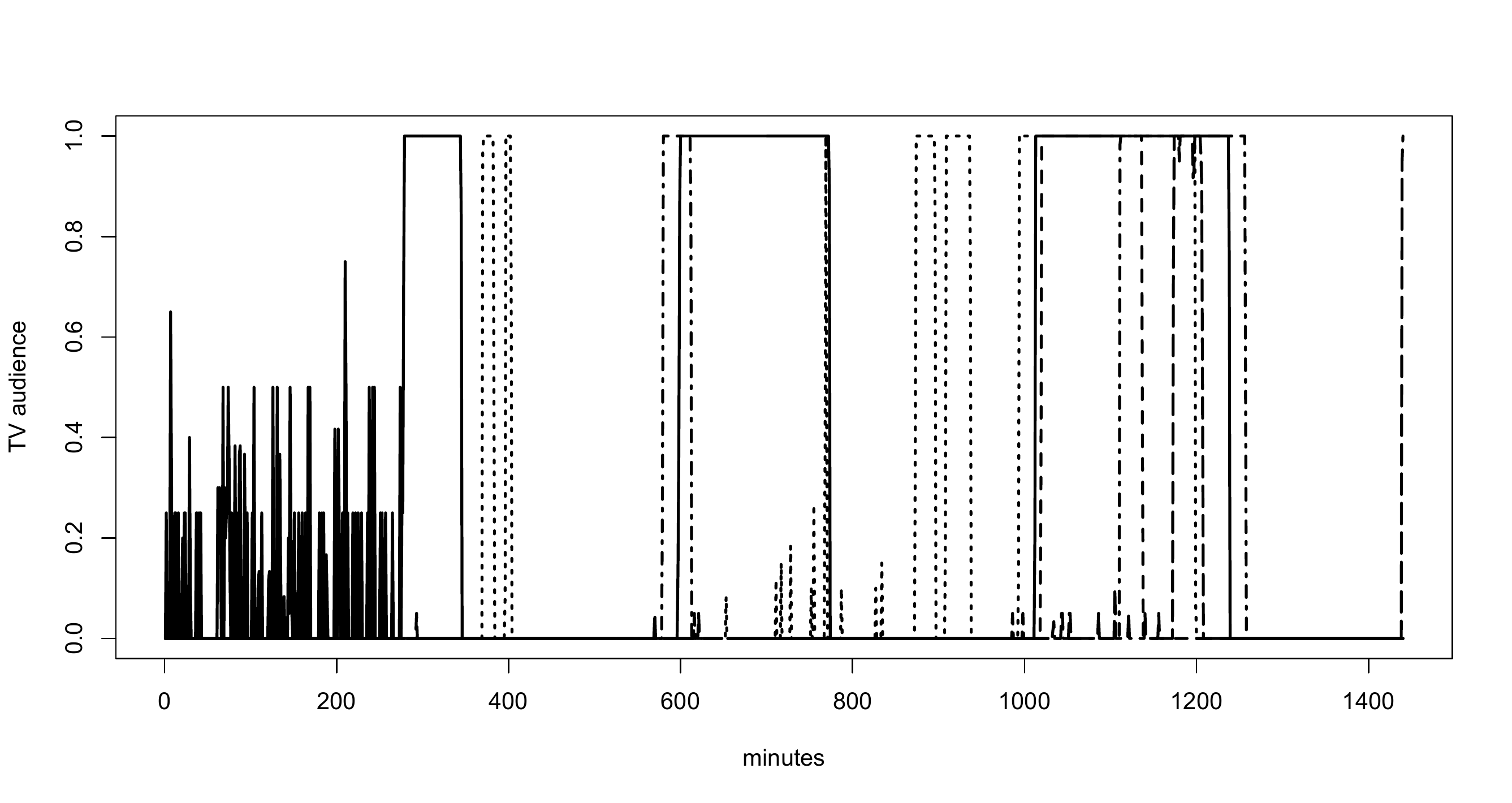} 
\caption{A sample of 5 observations of individual audience profiles measured every minute over a period of 24 hours.}
\label{fig3}
\end{center}
\end{figure}

We are interested in building profiles of the evolution along time of the total audience for people who watched television at least one minute on the 6th September 2010. About 1600 people, among the initial sample whose size is around 7000,  did not  watch television at all  this day, so that we finally get a sample of $n=5422$  individual audiences, aggregated along all television channels and measured every minute over a period of 24 hours. An  observation $Z_i$ is a vector belonging to $[0,1]^d,$ with $d=1440,$ each component giving the fraction of time spent watching television during the corresponding minute of the day.  A sample of 5 individual temporal profiles is drawn in Figure \ref{fig3}.
Clustering techniques based on medoids and representative elements of the sample (\textit{e.g.} PAM, CLARA and CLARANS) are not really interesting in this context since they will return centers of the form of the profiles drawn in Figure \ref{fig3} which are, in great majority, constituted of 0 and 1 and are consequently difficult to interpret.  Furthermore, the dimension being very large, $d=1440,$ these algorithms which do not consider all the dimensions of the data, as seen in the simulation study, will lead to a minimum value of the empirical $L_1$  error (\ref{def:emprisk}) that will be substantially larger than for the $k$-means and our recursive $k$-medians. Indeed, at the optimum,  the value of the $L_1$ empirical error is 0.2455 for the $k$-medians, 0.2471 for the $k$-means and  0.2692 for PAM.

The cluster centers, estimated with our averaged algorithm for $k=5,$ with a parameter value selected automatically, $c=0.2471,$ and 100 different starting points, are drawn in Figure \ref{fig4}. They have been ordered in decreasing order according to the sizes of the partitions and labelled  Cl.1 to Cl.5. Cluster 1 (Cl.1) is thus the  largest cluster and it contains about 35\% of the elements of the sample. It corresponds to individuals that do not watch television much during the day, with a cumulative  audience of about 42 minutes. At the opposite, cluster 5, which represents about 12\% of the sample, is associated to high audience rates during nearly all the day with a cumulative audience of about 592 minutes. Clusters 2, 3 and 4 correspond to intermediate consumption levels and can be distinguished according to whether  the audience occurs during the evening or at night. For example Cluster 4, which represents  16\% of the sample, is related to people watching television late at night, with a cumulative audience of about 310 minutes.

 \begin{figure}[htb]
\begin{center}
\includegraphics[height=12cm,width=14cm]{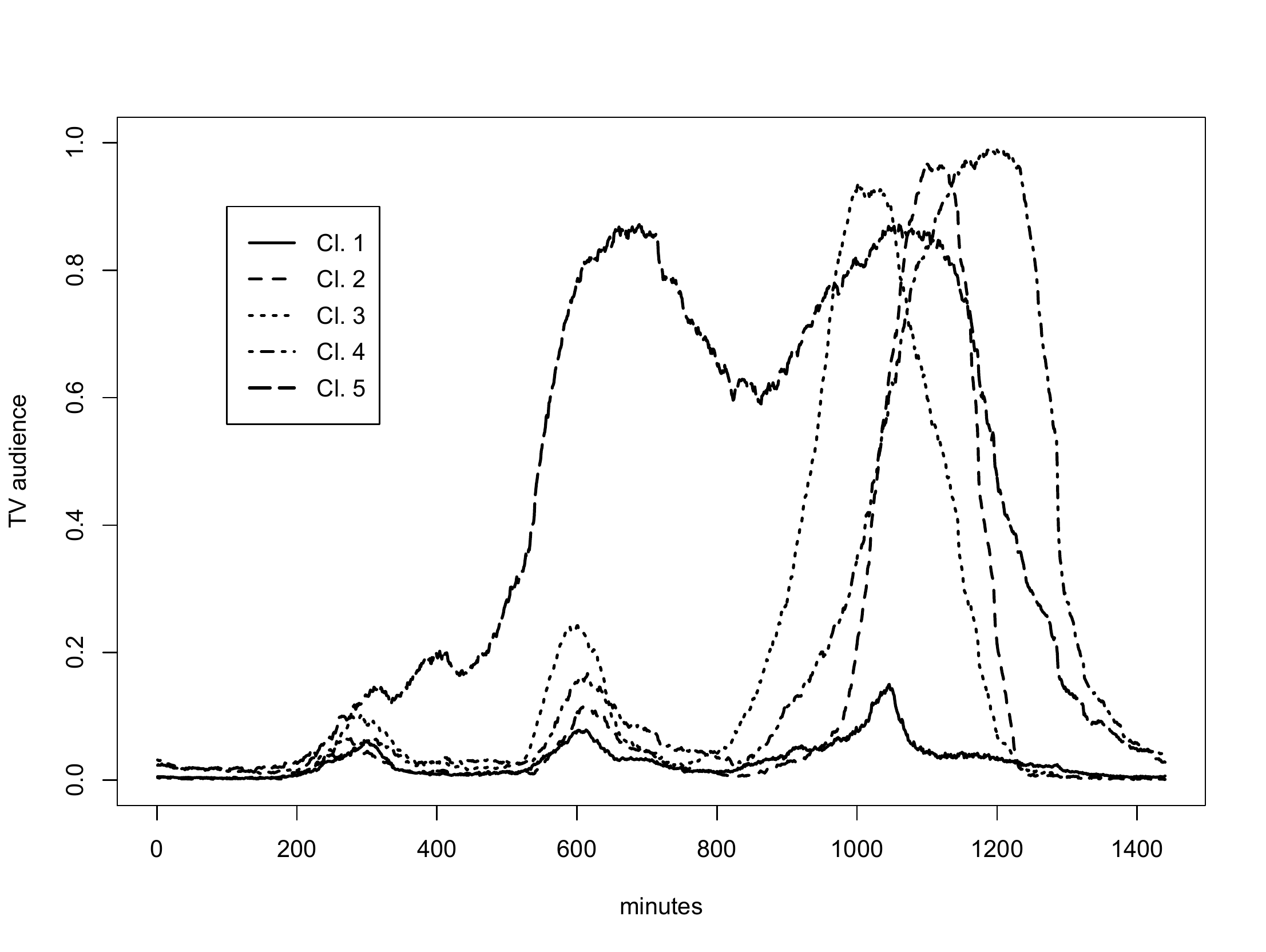} 
\caption{Cluster centers for temporal television audience profiles measured every minute over a period of 24 hours.}
\label{fig4}
\end{center}
\end{figure}

\clearpage
\section*{Appendix : Proof of Theorem~\ref{thm-principal}}

The proof of Theorem~\ref{thm-principal} relies on the following light version of the main theorem in \cite{Monnez}, section 2.1. The proof of Theorem~\ref{thm-principal} consists in checking that all the conditions of the following theorem are satisfied.

\begin{thm}[Monnez (2006)] 
Assuming
\begin{itemize}
\item[$(A1a)$] $g$ is a non negative function;
\item[$(A1b)$] There exists a constant $L>0$ such that, for all $n \geq 1$,
\[g(X_{n+1})-g(X_n)\leq \langle X_{n+1}-X_n,\nabla g(X_n)\rangle+L\nrm{X_{n+1}-X_n}^2 \quad \mbox{a.s.};\]
\item[$(A1c)$] The sequence $(X_n)$ is almost surely bounded and $\nabla g$ is continuous almost everywhere on the compact set containing $(X_n)$;
\item[$(A2)$] There exists four sequences of random variables $(B_n)$, $(C_n)$,$(D_n)$ and $(E_n)$ in $\Rset^{+}$ adapted to the sequence $(\CF_n)$ such that a.s.:
\item[$(A2a)$] $\nrm{\sqrt{A_n}\E[V_n|\CF_n]}^2\leq B_n g(X_n)+C_n$ and $\sum_{n=1}^{\infty}(B_n+C_n)<\infty$;
\item[$(A2b)$] $\E[\nrm{A_nV_n}^2|\CF_n]\leq D_n g(X_n)+E_n$ and $\sum_{n=1}^{\infty}(D_n+E_n)<\infty$;
\item[$(A3)$] $\sup_n a_n^r < \min(\frac12, \frac{1}{4L})$ a.s., $\sum_{n=1}^{\infty}\max_r a_n^r=\infty$ a.s. and
$$ \sup_{n}\frac{\max_ra_n^r}{\min_r a_n^r}<\infty \quad \mbox{a.s.}$$
then the distance of $X_n$ to the set of stationary points of $g$ converges almost surely to zero.
\end{itemize}
\end{thm}

\begin{proof}[Proof of Theorem~\ref{thm-principal}] 
\ \\

Let us now check that all the conditions in Theorem 2 are fulfilled in our context.

\textbf{Step 1: proof of $(A1b)$}\\ 

Let $A=X_n$ and $B=X_{n+1}$. Since $X_n$ is absolutely continuous with respect to Lebesgue measure, $\sum_{r=1}^kI_r(Z;A)=1$ a.s. and one gets
$$g(B)=\E\left[\min_r \nrm{Z-B^r}\right]=\E\left[\sum_{r=1}^kI_r(Z;A)\min_j \nrm{Z-B^j}\right],$$
and it comes
$$g(B)\leq \sum_{r=1}^k\E\left[I_r(Z;A)\nrm{Z-B^r}\right],$$
which yields
$$g(B)-g(A)\leq \sum_{r=1}^k \E\left[I_r(Z;A)\left(\nrm{Z-B^r}-\nrm{Z-A^r}\right)\right].$$
The application $x\mapsto \nrm{z-x^r}$ is a continuous function whose gradient 
\[\nabla_r \nrm{z-x^r}=\frac{x^r-z}{\nrm{x^r-z}}\]
is also continuous for $x^r \neq z$. Then almost surely for $d\geq 2$, there exists $C^r=A^r+\mu^r(B^r-A^r)$, $0\leq \mu^r\leq 1$, such that
\[\nrm{Z-B^r}-\nrm{Z-A^r}=\langle B^r-A^r, \nabla_r \nrm{Z-C^r}\rangle.\]
Consequently for all $d\geq 2$,
\[g(B)-g(A) \leq \sum_{r=1}^k\E\left[I_r(Z;A)\langle B^r-A^r,\nabla_r\nrm{Z-C^r}\rangle\right],\]
so that
\begin{eqnarray*}
g(B)-g(A)&\leq& \sum_{r=1}^k\E\left[I_r(Z;A)\langle B^r-A^r,\nabla_r\nrm{Z-C^r}-\nabla_r\nrm{Z-A^r}\rangle\right]\\
&+& \sum_{r=1}^k\E\left[I_r(Z;A)\langle B^r-A^r,\nabla_r\nrm{Z-A^r}\rangle\right] \egaldef (1)+(2) 
\end{eqnarray*}
On the one hand
\[(2)=\sum_{r=1}^k\langle B^r-A^r,\nabla_rg(A)\rangle=\langle B-A,\nabla g(A)\rangle,\]
and on the other hand
\[(1) \leq \sum_{r=1}^k\nrm{B^r-A^r}\E\left[\nrm{\nabla_r \nrm{Z-C^r}-\nabla_r\nrm{Z-A^r}}\right],\]
hence since
\[\nrm{\nabla_r\nrm{Z-C^r}-\nabla_r\nrm{Z-A^r}}=\nrm{\frac{C^r-Z}{\nrm{C^r-Z}}-\frac{A^r-Z}{\nrm{A^r-Z}}}\leq 2 \frac{\nrm{C^r-A^r}}{\nrm{A^r-Z}},\]
one gets, with (H1c)
\[(1)\leq 2 \sum_{r=1}^k\nrm{B^r-A^r}\nrm{C^r-A^r}\E\left[\frac{1}{\nrm{Z-A^r}}\right] \leq 2C\sum_{r=1}^k\nrm{B^r-A^r}^2=2C \nrm{B-A}^2.\]
Consequently, we have
\[g(B)-g(A)\leq \langle B-A, \nabla g(A)\rangle + 2C \nrm{B-A}^2.\]

\medskip

\textbf{Step 2: Proof of the assertion: $\forall n \geq 1$, for all $r=1,...k$, $\nrm{X_n^r} \leq K+2 \sup_n a_n^r$}
\\ 

Let us prove by induction on $n$ that for all $n\in \Nset^{*}$, for all $r=1, \ldots, k$, $\nrm{X_n^r}\leq K+ 2 \sup_{n}a_n^r$. This inequality is trivial for the case $n=1$: $\nrm{X_1^r}\leq K$. Let $n\in \Nset^{*}$ such that $\nrm{X_n^r}\leq K+ 2 \sup_{n}a_n^r$, $\forall r \in \{1, \ldots, k\}$. Let $r\in \{1, \ldots, k\}$. First we assume that $\nrm{X_n^r}\leq K+ a_n^r$. Then it comes 
\[\nrm{X_{n+1}^r}\leq \nrm{X_{n}^r}+a_n^rI_r(Z_n;X_n)\leq \nrm{X_{n}^r}+a_n^r\leq K +2a_n^r.\]
Now in the case when $K+a_n^r<\nrm{X_n^r}\leq K+ 2 \sup_{n}a_n^r$, one gets
\[\nrm{X_{n}^r}>K+a_n^r\geq \nrm{Z_n}+a_n^r,\] 
and then 
\[\nrm{X_n^r-Z_n}\geq \abs{\nrm{X_n^r}-\nrm{Z_n}}>a_n^r.\]
Since for $I_r(Z_n;X_n)=0$, $X_{n+1}^r=X_n^r$, it remains to deal with the unique index $r$ such that $I_r(Z_n;X_n)=1$. In that case,
\[X_{n+1}^r=X_n^r-a_n^r\frac{X_n^r-Z_n}{\nrm{X_n^r-Z_n}}=\left(1-\frac{a_n^r}{\nrm{X_n^r-Z_n}}\right)X_n^r+a_n^r\frac{Z_n}{\nrm{X_n^r-Z_n}}.\]
By (H1b) and from the inequalities $a_n^r/\nrm{X_n^r-Z_n}<1$ and $\nrm{Z_n}\leq K < \nrm{X_n^r}$, we have, 
\[\nrm{X_{n+1}^r}< \left(1-\frac{a_n^r}{\nrm{X_n^r-Z_n}}\right)\nrm{X_n^r}+a_n^r\frac{\nrm{X_n^r}}{\nrm{X_n^r-Z_n}}=\nrm{X_n^r},\]
which leads to $\nrm{X_{n+1}^r}\leq K+2\sup_{n}a_n^r$ and concludes the proof by induction.\\

\medskip

\textbf{Step 3: Proof of $(A1c)$}\\

From the integral form
\[\frac{\partial g}{\partial x^{r}_j}(x)=\int_{\Rset^d\setminus \{x^r\}}I_r(z;x)\frac{x^{r}_j-z_j}{\nrm{z-x^r}}f(z)dz,\]
it is easy to see that $\frac{\partial g}{\partial x^{r}_j}$ is a continuous function of $x$.

\medskip
\textbf{Step 4: Proof of $(A2a)$}\\

\noindent The definition of $V_n^r$ implies that $\E[V_n^r|\CF_n]=0$ and hence  $\E[V_n|\CF_n]=0$.

\medskip

\textbf{Step 5: Proof of $(A2b)$}\\
\begin{eqnarray*}
\E\left[\nrm{A_nV_n}^2|\CF_n\right]&=&\sum_{r=1}^k\E\left[\left(a_n^r\right)^2\nrm{V_n^r}^2|\CF_n\right]\\
&\leq & \sum_{r=1}^k\left(a_n^r\right)^2\E\left[I_r(Z_n;X_n)\frac{\nrm{X_n^r-Z_n}^2}{\nrm{X_n^r-Z_n}^2}\Big|\CF_n\right]\\
&\leq & \sum_{r=1}^k(a_n^r)^2.
\end{eqnarray*}
Hence assuming (H3), one gets
\[\E\left[\sum_{n=1}^{\infty}\E\left[\nrm{A_nV_n}^2|\CF_n\right]\right] < \infty.\]
In the case when (H3') holds instead of (H3), one has
\[\E\left[\sum_{n=1}^{\infty}\E\left[\nrm{A_nV_n}^2|\CF_n\right]\right]\leq \sum_{n=1}^{\infty}\sum_{r=1}^k\E\left[(a_n^r)^2I_r(Z_n;X_n)\right]<\infty.\]
Consequently,
\[\sum_{n=1}^{\infty}\E\left[\nrm{A_nV_n}^2|\CF_n\right]<\infty \quad \mbox{a.s},\]
which concludes the proof.% of Theorem~\ref{thm-principal}
\end{proof}

\medskip

\noindent \textbf{Acknowledgements.} We thank the anonymous referees for their valuable suggestions. We also thank the  M\'ediam\'etrie company for allowing us to illustrate our sequential clustering technique  with their data.

%\section*{References}

\end{document}